\documentclass[a4paper,UKenglish,cleveref, autoref, thm-restate]{lipics-v2021}

\usepackage{graphicx}
\usepackage{amsmath}
\usepackage{amssymb}
\usepackage[group-separator={,}]{siunitx}
\usepackage[ruled,vlined]{algorithm2e}
\SetKwData{Rank}{rank}
\SetKwData{Select}{select}
\SetKwData{BWT}{BWT}
\SetKwData{SA}{SA}
\SetKwData{MS}{MS}
\SetKwData{MS}{MS}
\SetKwData{SLP}{SLP}
\SetKwData{MEMs}{MEMs}
\SetKwData{Plen}{plen}
\SetKwData{Slen}{slen}
\def\Pos{\mbox{\rm {\sf pos}}}
\def\Len{\mbox{\rm {\sf len}}}
\SetKwData{RemoveDuplicates}{removeDuplicates}
\SetKwData{SortByPosition}{sortByPosition}
\SetKwData{Overlap}{overlap}
\SetKwData{Unique}{unique}
\SetKwData{Ppos}{Ppos}
\SetKwData{Tpos}{Tpos}
\SetKwData{Prevpos}{PrevPos}
\SetKwData{Follpos}{FollPos}

\SetKwData{Twice}{slen}
\SetKwData{LF}{LF}
\SetKwData{pred}{pred}
\SetKwData{F}{F}
\SetKwData{Min}{min}
\SetKwData{Max}{max}
\SetKwData{LCE}{LCE}
\SetKwData{ISA}{ISA}
\SetKwData{LCP}{LCP}
\SetKwData{add}{add}

\SetKwInOut{Input}{Input}
\SetKwInOut{Output}{Output}
\usepackage{algpseudocode}
\algnewcommand\And{\bf{and }}
\algnewcommand\Or{\bf{or }}
\algnewcommand\False{\bf{false }}
\algnewcommand\True{\bf{true }}
\algnewcommand\To{\bf{to }}
\algnewcommand\Down{\bf{down }}

\usepackage{xspace}
\newcommand{\Oh}{{\cal O}}
\newcommand{\Lh}{{\cal L}}
\newcommand{\eMS}{{\sf eMS}\xspace}
\newcommand{\MUMs}{{\sf MUMs}\xspace}
\newcommand{\MUM}{{\sf MUM}\xspace}

\usepackage{wrapfig}

\def\ours{\mbox{\rm {\sc mum-phinder}}\xspace}

\newlength\mylen
\newcolumntype{C}{>{\hfil$}p{\mylen}<{$\hfil}} %

\title{Computing Maximal Unique Matches with the  $r$-index}

\author{Sara Giuliani}{Department of Computer Science, University of Verona, Italy}{sara.giuliani\_01@univr.it}{https://orcid.org/0000-0002-1179-3929}{}

\author{Giuseppe Romana}{Department of Computer Science, University of Palermo, Italy}{giuseppe.romana01@unipa.it}{https://orcid.org/0000-0002-3489-0684
}{}

\author{Massimiliano Rossi}{Department of Computer and Information Science and Engineering, University of Florida, USA}{rossi.m@ufl.edu}{https://orcid.org/0000-0002-3012-1394}{National Science Foundation NSF EAGER (Grant No. 2118251), and National Institutes of Health (NIH) NIAID (Grant No. HG011392).}

\authorrunning{S. Giuliani, G. Romana, and M. Rossi} %

\Copyright{Sara Giuliani, Giuseppe Romana, and Massimiliano Rossi} %



\category{} %

\relatedversion{} %

\supplementdetails{Software}{https://github.com/saragiuliani/mum-phinder}

\funding{}%

\acknowledgements{We thank Travis Gagie for suggesting this problem as a project for his course CSCI 6905 at Dalhousie University. We also thank the anonymous reviewers for their insightful comments.}%

\nolinenumbers %

\hideLIPIcs  %

\begin{document}

\maketitle      %

\begin{abstract}

In recent years, pangenomes received increasing attention from the scientific community for their ability to incorporate population variation information and alleviate reference genome bias. Maximal Exact Matches (\MEMs) and Maximal Unique Matches (\MUMs) have proven themselves to be useful in multiple bioinformatic contexts, for example short-read alignment and multiple-genome alignment. However, standard techniques using suffix trees and FM-indexes do not scale to a pangenomic level. Recently, Gagie et al. [JACM 20] introduced the $r$-index that is a Burrows-Wheeler Transform (\BWT)-based index able to handle hundreds of human genomes. Later, Rossi et al. [JCB 22] enabled the computation of \MEMs using the $r$-index, and Boucher et al. [DCC 21] showed how to compute them in a streaming fashion.

In this paper, we show how to augment Boucher et al.'s approach to enable the computation of \MUMs on the $r$-index, while preserving the space and time bounds. We add additional $\Oh(r)$ samples of the longest common prefix (\LCP) array, where $r$ is the number of equal-letter runs of the \BWT, that permits the computation of the second longest match of the pattern suffix with respect to the input text, which in turn allows the computation of candidate \MUMs. 
We implemented a proof-of-concept of our approach, that we call \ours, and tested on real-world datasets. We compared our approach with competing methods that are able to compute \MUMs. We observe that our method is up to 8 times smaller, while up to 19 times slower when the dataset is not highly repetitive, while on highly repetitive data, our method is up to 6.5 times slower and uses up to 25 times less memory.

\end{abstract}

\section{Introduction}

With the advent of third-generation sequencing, the quality of assembled genomes drastically increased. In the last year the Telomere-to-Telomere project released the first complete haploid human genome~\cite{Nurk_2021-og} and the Human Pangenome Reference Consortium (HPRC) plans to release hundreds of high-quality assembled genomes to be used as a pangenome reference. One important step to enable the use of these high-quality assembled genomes is to build a multiple-sequence alignment of the genomes. Tools like MUMmer~\cite{Kurtz2004-cf,Marcais2018-ys}, and Mauve~\cite{Darling2004-en} proposed a solution to the original problem of multiple-sequence alignment by using Maximal Unique Matches (MUMs) between two input sequences as prospective anchors for an alignment. MUMs are long stretches of the genomes that are equal in both genomes and occur only once in each of them. To reduce the computational costs of computing the MUMs, progressive approaches have also been developed like progressive Mauve~\cite{Darling2010-jw} and progressive Cactus~\cite{Armstrong2020-fi} that enables the construction of pangenome graphs, among others, that have been used in recent aligners like Giraffe~\cite{Siren2021-oo}. MUMs have also been proven useful for strain level read quantification~\cite{Zhu2020-hs}, and as a computationally efficient genomic distance measure~\cite{Deloger2009-ij}.

Recent advances in pangenomics~\cite{Rossi2022-ah, DBLP:conf/dcc/BoucherGIKLMNP021} demonstrated that it is possible to index hundreds of Human Genomes and to query such an index to find supersets of MUMs that are maximal exact matches (MEMs), which are substrings of the pattern that occur in the reference and that cannot be extended neither on the left nor on the right. The tool called MONI~\cite{Rossi2022-ah} requires two passes over the query sequence to report the MEMs. Later PHONI~\cite{DBLP:conf/dcc/BoucherGIKLMNP021} showed how to modify the query to compute the MEMs in a streaming fashion, with only one single pass over the query string. Both MONI and PHONI are built on top of an $r$-index~\cite{GagieNP20} and a straight-line program SLP~\cite{Gagie2020-lg}. Their main objective is to compute the so called matching statistics (see Definition~\ref{def:ms}) of the pattern with respect to the text, that can be used to compute the MEMs with a linear scan. While, MONI uses the SLP for random access to the text, and needs to store additional information to compute the matching statistics and the MEMs, PHONI uses the SLP to compute efficient longest common extension (LCE) queries which allow to compute the matching statistics and the MEMs with only one scan of the query. 

We present \ours, a tool that is able to compute MUMs of a query pattern against an index on a commodity computer. The main observation of our approach is to extend the definition of matching statistics to include, for each suffix of the pattern, the information of the length of the second longest match of the suffix in the text, which allows to decide whether a MEM is also unique. We extended PHONI to keep track at each step of the query, the second longest match of the pattern in the index, and its length. To do this, we add $O(r)$ samples of the longest common prefix (\LCP) array to PHONI. 

We evaluated our algorithm on real-world datasets, and we tested \ours against MUMmer~\cite{Marcais2018-ys}. We measured time and memory required by both tools for sets of increasing size of haplotypes of human chromosome 19 and SARS-CoV2 genomes and queried using one haplotype of chromosome 19 and one SARS-CoV2 genome not present in the dataset. We report that \ours requires consistently less memory than MUMer for all experiments being up to 25 times smaller. Although MUMer is generally faster than ours (18 times faster for 1 haplotype of chromosome 19, and 6.5 times faster for 12,500 SARS-CoV2 genomes), it cannot process longer sequences due to memory limitations. Additionally, we observe that when increasing the number of sequences in the dataset, the construction time of \ours\ increases, while the query time decreases. This phenomenon is due to the increase in the number of matches in the search process, that prevents the use of more computational-demanding operations. 
Note that, due to the use of the $r$-index, the efficiency of our method increases when the dataset is highly repetitive as in the case of pangenomes. 

\section{Preliminaries}
Let $\Sigma = \{a_0 < a_1 < \ldots < a_{\sigma-1}\}$ be an \emph{ordered alphabet}, where $<$ represents the lexicographical order.
A \emph{string} (or \emph{text}) $T$ is a sequence of characters $T[0] T[1] \cdots T[n-1]$ such that $T[j] \in \Sigma$ for all $j \in [0..n)$.
The length of a string is denoted by $|T|$.
We refer to the empty string with $\varepsilon$, that is the only substring of length $0$.

We denote a factor (or substring) of $T$ as $T[i..j) = T[i] T[i+1] \cdots T[j-1]$ if $i < j$, and $T[i..j) = \varepsilon$ otherwise. We refer to $T[0..j)$ as the $j-1$-th prefix of $T$ and to $T[i..n)$ as the $i$-th suffix of $T$.

We assume throughout the paper that the text $T$ is terminated by termination character \$ that does not occur in the original text and it is lexicographically smaller than all the other characters in the alphabet.

\paragraph*{Suffix array, inverse suffix array, and longest common prefix array}
The \emph{Suffix array} (\SA) of a string $T[0..n)$ is an array of length $n$ such that $T[\SA[i]..n) < T[\SA[j]..n)$ for any $0 \leq i < j < n$.
The \emph{Inverse Suffix array} (\ISA) is the inverse of \SA, i.e.\ $\ISA[i] = j$ if and only if $\SA[j] = i$.
Let $lcp(u,v)$ be the length of the longest common prefix between two strings $u$ and $v$, that is $u[0..lcp(u,v)) = v[0..lcp(u,v))$ but $u[lcp(u,v)] \neq v[lcp(u,v)]$ (assuming $lcp(u,v) < \min\{|u|,|v|\}$).
The \emph{Longest Common Prefix array} (\LCP) of $T[0..n)$ is an array of length $n$ such that $\LCP[0] = 0$ and $\LCP[i] = lcp(T[\SA[i-1]..n), T[\SA[i]..n))$, for any $0< i < n$.

\paragraph*{Burrows-Wheeler Transform, Run-Length Encoding, and $r$-index}
The \emph{Burrows-Wheeler Transform} (\BWT) of $T$ is a reversible transformation of the characters of $T$~\cite{Bwt}.
That is the concatenation of the characters preceding the suffixes of $T$ listed in lexicographic order, i.e., for all $0 \leq i < n$, $\BWT[i] = T[\SA[i]-1 \mod n]$. 
The $\LF$-\emph{mapping} is the function that maps every character in the \BWT with its preceding text character, in the \BWT, i.e. $\LF(i) = \ISA[\SA[i]-1 \mod n]$.

The \emph{run-length encoding} of a string $T$ is the representation of maximal equal-letter runs of $T$ as pairs $(c, \ell)$, where $c$ is the letter of the run and $\ell > 0$ is the length of the run. For example, the run length encoding of $T = AAACAAGGGG$ is $(A,3)(C,1)(A,2)(G,4)$.
We refer to the number of runs of the \BWT\ with $r$.

The $\BWT$ tends to create long equal-letter runs on highly repetitive texts such as genomic datasets. The run-length encoding applied to the \BWT (in short RLBWT) is the basis of many lossless data compressors and text indexes, such as the FM-index~\cite{DBLP:conf/focs/FerraginaM00} which is the base of widely used bioinformatics tools such as Bowtie~\cite{Bowtie} and BWA~\cite{bwa}.
Although the \BWT\ can be stored and queried in compressed space~\cite{Mkinen2005SuccinctSA}, the number of samples of the \SA\ required by the index grows with the length of the uncompressed text.
To overcome this issue Gagie et al.~\cite{GagieNP20} proposed the $r$\emph{-index} whose number of \SA\ samples grows with the number of runs $r$ of the \BWT. 
The $r$-index is a text index composed by the run-length encoded \BWT\ and the \SA\ sampled at run boundaries, i.e., in correspondence of the first and last character of a run of the \BWT, and it is able to retrieve the missing values of the \SA\ by using a predecessor data structure on the samples of the \SA.

\paragraph*{Grammar and straight-line program}
A \emph{context-free grammar} $\mathcal{G} = \{V, \Sigma, R, S\}$ consists in a set of \emph{variables} $V$, a set of \emph{terminal symbols} $\Sigma$, a set of \emph{rules} $R$ of the type $A \mapsto \alpha$, where $A \in V$ and $\alpha \in \{V\cup \Sigma\}^*$, and the \emph{start variable} $S \in V$.
The \emph{language of the grammar} $\Lh(\mathcal{G}) \subseteq \Sigma^*$ is the set of all words over the alphabet of terminal symbols generated after applying some rules in $R$ starting from $S$.
When $\Lh(\mathcal{G})$ contains only one string $T$, that is $\mathcal{G}$ only generates $T$, then the grammar $\mathcal{G}$ is called \emph{straight-line program} (SLP).

\paragraph*{Longest Common Extension, \Rank, and \Select queries}
Given a text $T[0..n)$, the \emph{longest common extension} (\LCE) query between two positions $0\leq i,j <n$ in $T$ is the length of the longest common prefix of $T[i..n)$ and $T[j..n)$.
Thus, if $\ell = \LCE(i,j)$, then $T[i..i+\ell) = T[j..j+\ell)$ and either $T[i+\ell] \neq T[j+\ell]$ or either $i+\ell = n$ or $j+\ell = n$.

Given a character $c$ and an integer $i$, we define $T.\Rank_c(i)$ as the number of occurrences of the character $c$ in the prefix $T[0..i)$, while we define $T.\Select_c(i)$ as the position $p \in [0..n)$ of the $i$th occurrence of $c$ in $T$ if it exists, and $p = n$ otherwise.

\section{Computing MUMs using MS}

Given a text $T[0..n)$ and a pattern $P[0..m)$, we refer to any factor in $P$ that also occurs in $T$ as a \emph{match}.
A match $w$ in $P$ can be defined as a pair $(i, \ell)$ such that $w = P[i..i+\ell)$.
We say that $w$ is \emph{maximal} if the match can not be extended neither on the left nor on the right, i.e. either $i = 0$ or $P[i-1..i+\ell)$ does not occur in $T$ and either $i = m - \ell$ or $P[i..i+\ell+1)$ does not occur in $T$. 

\begin{definition}
Given a text $T$ and a pattern $P$, a Maximal Unique Match (\MUM) is a maximal match that occurs exactly once in $T$ and $P$.
\end{definition}

\begin{example}
Let $T$ =	ACACTCTTAC{\bf{ACC}}ATATCATCAA\$ be the text and $P$ = AACCTAA the pattern.
The factor {\bf{AA}} is maximal in $P$ and occurs only once in $T$, while it is repeated in $P$ at positions $0$ and $5$.
The factor {\bf{CT}} of $P$ starting in position 3 is a maximal match that occurs only once in $P$, but it is not unique in $T$.
The factor {\bf{CC}} of $P$ starting in position 2 is unique in both $T$ and $P$, but both can be extended on the left with an {\bf{A}}.
On the other hand, the factor $P[1..4)= T[10..13) = ${\bf{ACC}} is a \MUM. 
\end{example}

From now on, we refer to the set of all maximal unique matches between $T$ and $P$ as $\MUMs$.
In \cite{DBLP:conf/dcc/BoucherGIKLMNP021} the authors showed how to compute maximal matches (not necessarily unique neither in $T$ nor $P$) in $\mathcal{O}(r+g)$ space, where $r$ is the number of runs of the \BWT of $T$ and $g$ is the size of the \SLP representing the text $T$. 
This is achieved by computing the \emph{matching statistics}, for which we report the definition given in \cite{DBLP:conf/dcc/BoucherGIKLMNP021}.

\begin{definition}[\cite{DBLP:conf/dcc/BoucherGIKLMNP021}]
\label{def:ms}
The matching statistics \MS of a pattern $P[0..m)$ with respect to a text $T[0..n)$ is an array of (position, length)-pairs $\MS[0..m)$ such that
\begin{itemize}
    \item $P[i..i + \MS[i].\Len) = T[\MS[i].\Pos..\MS[i].\Pos + \MS[i].\Len )$;
    \item either $i = m-\MS[i].\Len$ or $P[i..i + \MS[i].\Len +1)$ does not occur in $T$.
\end{itemize}
That is, $\MS[i].\Pos$ is the starting position in $T$ of an occurrence of the longest prefix of $P[i..m)$ that occurs in $T$, and $\MS[i].\Len$ is its length.
\end{definition}

A known property of the matching statistics is that for all $i > 0$, $\MS[i].\Len \geq \MS[i-1].\Len -1$.

Our objective is to show how to further compute \MUMs\ within the same space bound. 
For our purpose, we extend the definition of $\MS$ array with an additional information field to each entry.

\begin{definition}\label{def:ems}
 Given a text $T =[0 \ldots n)$ and a pattern $P =[0 \ldots m)$, we define the extended matching statistics $\eMS$ as an array of ($\Pos, \Len, \Twice$)-tuples $\eMS[0 \ldots m)$ such that
\begin{itemize}
    \item $\eMS[i].\Pos = \MS[i].\Pos$ and $\eMS[i].\Len = \MS[i].\Len$;
    \item $\eMS[i].\Twice$ is the largest value $\ell$ for which there exists $p \neq \eMS[i].\Pos$ such that $P[i..i+\ell) = T[p ..p +\ell)$.
\end{itemize}
In other words, $\eMS[i].\Twice$ is the length of the second longest match of a prefix $P[i..n)$ in $T$. 
\end{definition}
 
Note that $\eMS[i].\Twice \leq \eMS[i].\Len$, for any $i\in [0..m)$.

\subsection{Checking Maximality and Uniqueness of matches}

We now show how to compute $\MUMs$ by using the $\eMS$ array.
Lemma~\ref{le:unique_in_T} shows how to verify if a match occurs only once in $T$.

\begin{lemma}
\label{le:unique_in_T}
Given a text $T$, a pattern $P$, and the $\eMS$ array computed for $P$ with respect to $T$, let $w = P[i..i+\eMS[i].\Len) = T[\eMS[i].\Pos..\eMS[i].\Pos+\eMS[i].\Len)$ be a maximal match between a pattern $P[0..m)$ and a text $T[0..n)\$$.
Then $w$ occurs exactly once in $T$ if and only if $\eMS[i].\Twice < \eMS[i].\Len$.
\end{lemma}

\begin{proof}
For the \emph{if} direction, we assume by contradiction that $w$ is unique in $T$ and that $\eMS[i].\Twice \geq \eMS[i].\Len$. By definition, $\eMS[i].\Twice \leq \eMS[i].\Len$, hence we assume $\eMS[i].\Twice = \eMS[i].\Len$. By definition of $\eMS[i].\Twice$ there exists $p \neq \eMS[i].\Pos$ such that $w = P[i..i+\eMS[i].\Twice) = T[p..p+\eMS[i].\Twice) = T[\eMS[i].\Pos..\eMS[i].\Pos+\eMS[i].\Len)$, that contradicts the assumption that $w$ occurs only once in the text $T$.
Analogously, assume that $\eMS[i].\Twice < \eMS[i].\Len$ and that there exists a position $j \neq \eMS[i].\Pos$ such that $T[j..j+\eMS[i].\Len) = T[\eMS[i].\Pos..\eMS[i].\Pos+\eMS[i].\Len)$. However, this is in contradiction with the definition of $\eMS[i].\Twice$ and the assumption of $\eMS[i].\Twice < \eMS[i].\Len$, concluding the proof.
\end{proof}

We check the maximality of a match in the pattern using an analogous approach as in \cite{Rossi2022-ah}, that we summarize with the following lemma. 

\begin{lemma}
\label{le:maximality}
Given a text $T$, a pattern $P$, and the $\eMS$ array computed for $P$ with respect to $T$, let $w = P[i..i+\eMS[i].\Len)$ be a match with a text $T$.
Then $w$ is a maximal match if and only if either $i=0$ or $\eMS[i-1].\Len \leq \eMS[i].\Len$.
\end{lemma}

\begin{proof}
First we show that if $w = P[i..i+\eMS[i].\Len)$ is a maximal match then either $i=0$ or $\eMS[i-1].\Len \leq \eMS[i].\Len$. Let us assume that $w$ is not maximal and either $i=0$ or $\eMS[i-1].\Len \leq \eMS[i].\Len$, hence either $P[i..i+\eMS[i].\Len + 1)$ occurs in $T$ or $P[i-1..i+\eMS[i].\Len)$ occurs in $T$. The former case is in contradiction with the definition of \eMS, hence $P[i-1..i+\eMS[i].\Len)$ occurs in $T$. This implies that $i > 0$ and that $\eMS[i-1].\Len = \eMS[i].\Len + 1$ in contradiction with the hypothesis that $\eMS[i-1].\Len \leq \eMS[i].\Len$. 

Now we show that if either $i=0$ or $\eMS[i-1].\Len \leq \eMS[i].\Len$ then $w$ is a maximal match. By definition of $\eMS[i].\Len$, we know that either $i+\eMS[i].\Len = m$ or $P[i..i+\eMS[i].\Len+1)$ does not occur in $T\$$, that is $w$ cannot be extended on the right in $P$.
If $i = 0$ we can not further extend the match $w$ on the left, hence $w$ is maximal.
If $i>0$, then by definition of matching statistics it holds that $\eMS[i-1].\Len \leq \eMS[i].\Len + 1$. 
Note that if there exists a character $a \in \Sigma$ such that $P[i-1..i-1+\eMS[i-1].\Len) = aw$ and $aw$ occurs in $T$, then $\eMS[i-1] = \eMS[i]+1$. Hence if $\eMS[i-1] = \eMS[i]+1$ then it is easy to see that $w$ is not maximal because it can be extended on the left. It also follows that if $\eMS[i-1] \leq \eMS[i]$ then $w$ cannot be extended on the left, hence it is maximal and the thesis follows.
\end{proof}

Let $\Lh \subseteq[0..m)$ be the subset of positions in $P$ such that both Lemma \ref{le:unique_in_T} and Lemma \ref{le:maximality} hold, i.e. $\Lh$ contains all the positions in $P$ where a maximal match unique in $T$ starts.
One can notice that if a match $w_i = P[i..i+\eMS[i].\Len)$ is a \MUM, then $i \in \Lh$.

We first show that given $i \in \mathcal{L}$, if a match $w_i$ is not unique in $P$, then the second occurrence of $w_i$ in $P$ is contained in another maximal match unique in $T$.

\begin{lemma}
\label{lemma:twice in P}
Given a text $T$, a pattern $P$, and the $\eMS$ array computed for $P$ with respect to $T$, let $\mathcal{L}$ be the subset of positions in $P$ such that $w_i = P[i..i+ \eMS[i].\Len)$ is maximal and occurs only once in $T$ for all $i \in \mathcal{L}$.
Then, $w_i$ is not unique in $P$ if and only if there exist $i'\in \mathcal{L}\setminus\{i\}$ and two possibly empty strings $u,v$ such that $w_{i'} = u w_i v$ is a factor of $P$.
\end{lemma}

\begin{proof}
Let us assume by contradiction that such $i'$ does not exist, then let $j \notin \mathcal{L}$ be such that $P[j..j+|w_i|) = w_i$. Since $j \notin \mathcal{L}$ then either $P[j..j+|w_i|)$ is not unique in $T$, or it is not maximal. The former case it contradicts $i \in \mathcal{L}$ because $P[j..j+|w_i|) = w_i$ occurs twice in $T$. Hence, $P[j..j+|w_i|)$ occurs only once in $T$ and it is not maximal, therefore there exists $k \in \mathcal{L}$ such that $k \leq j$ and $|w_k| > |w_i|$ which contradict the hypothesis. The other direction of the proof is straightforward since by definition of $w_{i'}$, either $w_i$ occurs twice in $P$ or it is not maximal. 
\end{proof}

The following Lemma shows, for any $i \in \Lh$, if a match $w_i$ is unique in $P$ by using the $\eMS$ array.

\begin{lemma}
\label{le:unique_in_P}
Given a text $T$, a pattern $P$, and the $\eMS$ array computed for $P$ with respect to $T$, let $\mathcal{L}$ be the subset of positions in $P$ such that $w_i = P[i..i+ \eMS[i].\Len)$ is maximal and occurs only once in $T$, for all $i \in \mathcal{L}$.
Then, $w_i$ occurs only once in $P$ if and only if, for all $i' \in \mathcal{L} \setminus \{i\}$, either $\eMS[i].\Pos < \eMS[i'].\Pos$ or $\eMS[i].\Len + \eMS[i].\Pos > \eMS[i'].\Len + \eMS[i'].\Pos$.
\end{lemma}

\begin{proof}
We first show that if $w_i$ occurs only once in $P$ then for all $i' \in \mathcal{L} \setminus \{i\}$, either $\eMS[i].\Pos < \eMS[i'].\Pos$ or $\eMS[i].\Len + \eMS[i].\Pos > \eMS[i'].\Len + \eMS[i'].\Pos$.
Since $\mathcal{L}$ contains only positions of maximal matches unique in $T$, then for all for $i \in \mathcal{L}$ we can map $w_i$ to its occurrence in the text $T[\eMS[i].\Pos..\eMS[i].\Pos + \eMS[i].\Len)$.
Since $w_i$ occurs only once in $T$, by Lemma~\ref{lemma:twice in P} we have that $\eMS[i'].\Pos = \eMS[i].\Pos -|u|$ and $\eMS[i'].\Len = \eMS[i].\Len +|u| +|v|$.
Hence, $\eMS[i'].\Pos \leq \eMS[i].\Pos$ and $\eMS[i].\Pos + \eMS[i].\Len \leq \eMS[i'].\Pos + \eMS[i'].\Len$.

We now show the other direction of the implication. If given a position $i \in \mathcal{L}$ for all $i' \in \mathcal{L} \setminus \{i\}$, either $\eMS[i].\Pos < \eMS[i'].\Pos$ or $\eMS[i].\Len + \eMS[i].\Pos > \eMS[i'].\Len + \eMS[i'].\Pos$ then $w_i$ occurs only once in $P$. Assuming by contradiction that there exists a position $i \in \mathcal{L}$ such that for all $i' \in \mathcal{L} \setminus \{i\}$, either $\eMS[i].\Pos < \eMS[i'].\Pos$ or $\eMS[i].\Len + \eMS[i].\Pos > \eMS[i'].\Len + \eMS[i'].\Pos$ and $w_i$ does not occur only once in $P$, then by Lemma`\ref{lemma:twice in P} there exist $j\in \mathcal{L}$ and two possibly empty strings $u,v$ such that $w_{j} = u w_i v$ is a factor of $P$. It is easy to see that $\eMS[j].\Pos = \eMS[i].\Pos -|u|$ and $\eMS[j].\Len = \eMS[i].\Len +|u| +|v|$.
Hence, $\eMS[j].\Pos \leq \eMS[i].\Pos$ and $\eMS[i].\Pos + \eMS[i].\Len \leq \eMS[j].\Pos + \eMS[j].\Len$, in contradiction with the hypothesis, concluding the proof.
\end{proof}

We can summarize the previous Lemmas in the following Theorem.

\begin{theorem}
\label{th:mum}
Given a text $T$, a pattern $P$, and the $\eMS$ array computed for $P$ with respect to $T$, for all $0 \leq i < m$, $w_i = P[i..i+\eMS[i].\Len)$ is a \MUM if and only if $i \in \Lh$ and Lemma \ref{le:unique_in_P} holds.
\end{theorem}

\begin{example}
Let $T$ =	ACACTCTTAC{\bf{ACC}}ATATCATCAA\$ be the text and $P$ = AACCTAA the pattern.
In the table below we report the values of the \eMS of $P$ with respect to $T$. 

 \begin{center}
   	\settowidth\mylen{$10$}
    $
    \begin{array}{{r}|*{7}{C}}
    i & 0 & 1 & 2 & 3 & 4 & 5 & 6 \\
    \hline
    P[i] & \text{A} & \text{A} & \text{C} & \text{C} & \text{T} & \text{A} & \text{A} \\
    \eMS[i].\Pos & 21 & 10 & 11 & 5 & 6 & 21 & 8 \\
    \eMS[i].\Len & 2 & 3 & 2 & 2 & 2 & 2 & 1 \\
    \eMS[i].\Twice & 1 & 2 & 1 & 2 & 2 & 1 & 1 
    \end{array}
    $
  \end{center}

  It is easy to check that $\Lh = \{0, 1, 5\}$, where $\Lh$ contains those indices $i$ which verify both Lemma~\ref{le:unique_in_T} $(\eMS[i].\Twice < \eMS[i].\Len$) and Lemma~\ref{le:maximality} (either $i = 0$ or $\eMS[i-1].\Len \leq \eMS[i].\Len$).
  Note that $\eMS[0].\Pos = \eMS[5].\Pos$ and $\eMS[0].\Len = \eMS[5].\Len$, and by Lemma~\ref{le:unique_in_P} we know that $P[0..2)$($=P[5..7)$) is repeated in $P$.
  Since $\eMS[1].\Pos < \eMS[0].\Pos = \eMS[5].\Pos$, by Theorem~\ref{th:mum} the match $P[1..4) = T[10..13) = $ ACC is a MUM.
\end{example}

\subsection{Computing the second longest match}

Now we show how we can compute $\eMS$ extending the algorithm presented in Boucher et al.~\cite{DBLP:conf/dcc/BoucherGIKLMNP021} while preserving the same space-bound.

We can apply verbatim the algorithm of~\cite{DBLP:conf/dcc/BoucherGIKLMNP021} to compute the $\eMS[i].\Pos$ and $\eMS[i].\Len$ while we extend the algorithm to include the computation of $\eMS[i].\Slen$.
The following Lemma shows how to find the second longest match using the $\LCP$ array.

\begin{lemma}
\label{le:twice}
Given a text $T$, a pattern $P$, and the $\eMS$ array of $P$ with respect to $T$, let $P[i..i + \eMS[i].\Len) = T[\eMS[i].\Pos..\eMS[i].\Pos + \eMS[i].\Len)$ and $q = \ISA[\eMS[i].\Pos]$.
Then, for all $0\leq q < n$, $\eMS[i].\Twice = \max\{\LCP[q], \LCP[q +1]\}$,  where $\LCP[n]=0$.
\end{lemma}

\begin{proof}
Let us consider the set $\mathcal{T}=\{w_0 < w_1 < \ldots < w_n\}$ of the lexicographically sorted suffixes of $T$.
Then, for all $i \in [0..m)$, at least one suffix of $T$ starting with the second longest match $P[i..i+\eMS[i].\Twice)$ must be adjacent to $w_q = T[\eMS[i].\Pos..n)$ in $\mathcal{T}$.
Hence, assuming $q\neq 0$ and $q\neq n$, $\eMS[i].\Slen$ is either the $\LCP$ value between $w_{q-1}$ and $w_{q}$ or between $w_{q}$ and $w_{q+1}$, that are respectively $\LCP[q]$ and $\LCP[q+1]$.
Note that if $q = 0$ then both $\LCP[0]$ and $\LCP[1]$ exist, while for the case $q = n$ only $\LCP[n]$ is available, that is $\eMS[i].\Slen$ must be $\LCP[n]$.
\end{proof}

\section{Algorithm description}
In this section we present the algorithm that we use to compute $\MUMs$ that builds on the approach of Boucher et al.~\cite{DBLP:conf/dcc/BoucherGIKLMNP021} for the computation of the $\MS$ array. The authors showed how to use the $r$-index and the SLP of~\cite{BigRePair, Gagie2020-lg} to compute the $\MS$ array of a pattern $P[0..m)$ in $\Oh(m\cdot (t_{\LF} + t_{\LCE} + t_{\pred}))$ time, where $t_{\LF}$, $t_{\LCE}$, and $t_{\pred}$ represent the time to perform respectively one \LF, one $\LCE$, and one predecessor query.
Our algorithm extends Boucher et al.'s method by storing additional $\Oh(r)$ samples of the $\LCP$ array. Given a text $T[0..n)$ and a pattern $P[0..m)$, in the following, we first show how to compute the $\eMS$ array of $P$ with respect to $T$ using the $r$-index, the SLP, and the additional $\LCP$ array samples. Then we show how to apply Theorem~\ref{th:mum} to compute the \MUMs  \ from the $\eMS$ array.

\subsection{Computing the {\eMS} array}\label{ss:streaming_ems}

The key point of the algorithm is to extend the last computed match backwards when possible, otherwise we search for the new longest match that can be extended on the left by using the \BWT.
Let $q$ be the index such that $P[i..i+\eMS[i].\Len) = T[\SA[q]..\SA[q]+\eMS[i].\Len)$ is the longest match found at step $i$:
\begin{itemize}
    \item if $\BWT[q] = P[i-1]$, then it can be extended on the left, i.e. $P[i-1..i+\eMS[i].\Len) = T[\SA[q]-1..\SA[q]+\eMS[i].\Len)$;
    \item otherwise, we want to find the longest prefix of $P[i..i+\eMS[i].\Len)$ that is preceded by $P[i-1]$ in the text $T$. As observed in Bannai et al.~\cite{bannai2020refining} it can be either the suffix corresponding to the occurrence of $P[i-1]$ in the $\BWT$ immediately preceding or immediately following $q$, that we refer to as $q_p$ and $q_s$ respectively. Formally, $q_p = \max \{ j < q \mid \BWT[j] = P[i-1] \}$ and $q_s = \min \{ j > q \mid \BWT[j] = P[i-1] \}$.
\end{itemize} 

The algorithm to compute the $\Pos$ and $\Len$ entry of the $\eMS$ array is analogous to the procedure detailed in \cite{DBLP:conf/dcc/BoucherGIKLMNP021}.
We use the same data structures as the one defined in \cite{DBLP:conf/dcc/BoucherGIKLMNP021}, that are the run-length encoded $\BWT$ and the samples of the $\SA$ in correspondence of positions $q$ such that $\BWT[q]$ is either the first or the last symbol of an equal-letter run of the $\BWT$.
Note that both $q_p$ and $q_s$ are respectively the last and the first index of their corresponding equal-letter run.

An analogous reasoning can be formulated to compute the second longest match.

\begin{lemma}
\label{le:lcp_lce+1}
Given a text $T[0..n)$, let \LCP, \SA and \ISA be respectively the longest common prefix array, suffix array and inverse suffix array of $T$.
Then, for all $0 < q \leq n$, let $i,j$ be two integers such that $q-1 = \LF[i]$ and $q = \LF[j]$, then if $\BWT[i] \neq \BWT[j]$ then $\LCP[q] = 0$, otherwise $\LCP[q] = \LCE(\SA[i], \SA[j]) + 1$.
\end{lemma}

\begin{proof}
Let $w_q$ be the $q$-th suffix in lexicographic order.
Note that if $w_q = \$$ then $\LCP[q] = \LCP[q+1] = 0$.
For all $1 \leq q < n$, if $w_{q-1} = au\$$ and $w_{q} = bv\$$ for some $a < b \in \Sigma$ and some strings $u$ and $v$, then $\LCP[q] = 0$.
On the other hand, if $w_{q-1} = au\$$ and $w_{q} = av\$$, then $\LCP[q] = 1 + lcp(u\$, v\$)$.
The thesis follows by observing that the suffixes $u\$$ and $v\$$ respectively correspond to $w_{i}$ and $w_{j}$.
\end{proof}

\begin{wrapfigure}[21]{hr}{0.5\textwidth}
\begin{minipage}{\linewidth}
\includegraphics[width=6.3cm]{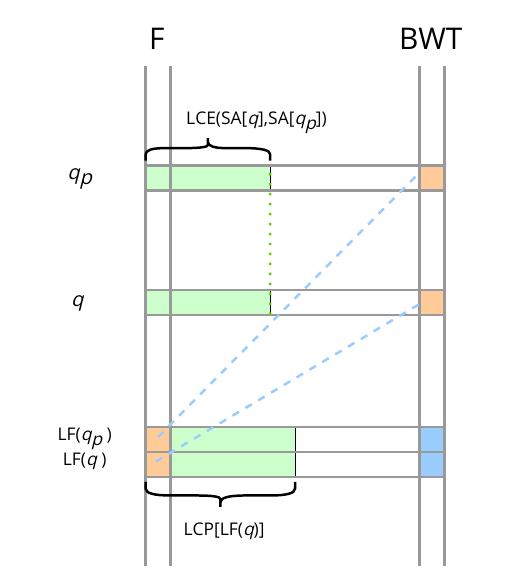}
\end{minipage}
\caption{
Application of Lemma \ref{le:lcp_lce+1} to compute $\LCP[\LF(q)]$ by extending the result of the last $\LCE$ query.  \label{fig:algoMSmatch_ll34}}
\end{wrapfigure}

Note that, the second longest match can be retrieved from the $\LCP$ values in correspondence of the longest maximal match (Lemma \ref{le:twice}).
Once we have the maximal match in position $q$ in the \BWT, we can compute $\LCP[q]$ and $\LCP[q+1]$ from the $\LCE$ queries on $T[\SA[q]..n)$ with $T[\SA[q_p]..n)$ and $T[\SA[q_s]..n)$ (Lemma \ref{le:lcp_lce+1}).

Moreover, assuming the index $q_p$ is the greatest index smaller than $q$ such that $\BWT[q_p] = \BWT[q]$, then $\LF(q_p) = \LF(q)-1$.
It follows that if $\BWT[\LF(q_p)] = \BWT[\LF(q)-1] = \BWT[\LF(q)]$, then $\LCP[\LF(q)])$ is an extension of the $\LCE$ query computed between $\SA[q_p]$ and $\SA[q]$ (see Figure \ref{fig:algoMSmatch_ll34}).
Symmetrically, if $q_s$ is the smallest index greater than $q$ such that $\BWT[q_s] = \BWT[q]$, then $\LF(q_s) = \LF(q) +1$.
Thus, at each iteration, we keep track of both $\LCP$ values computed to find the second longest match.

With respect to the implementation in \cite{DBLP:conf/dcc/BoucherGIKLMNP021}, we add $\Oh(r)$ sampled values from the $\LCP$ array.
Precisely, we store the $\LCP$ values between the first and the last two suffixes in correspondence of each equal-letter run (if only one suffix corresponds to a run we simply store 0).
As shown later, this allows to overcome the problem of computing the \LCE queries in case a position $p$ in $T$ is not stored in the sampled $\SA$, i.e. when $\ISA[p]$ is neither the first nor the last index of its equal-letter run.

For simplicity of exposition we ignore the cases when a \Select query of a symbol $c$ in the $\BWT$ fails.
However, whenever it happens, either $c$ does not occur in $T$ or we are attempting to find an occurrence out of the allowed range, that is between $0$ and the number of occurrences of the character $c$ minus $1$.
For the first case we can simply reset the algorithm starting from the next character of $P$ to process, while the second occurs when we are attempting to compute an $\LCE$ query, whose result can be safely set to $0$.

Algorithm~\ref{algo:MScomputation} computes the extended matching statistics $\eMS$ of the pattern $P = [0 \ldots m)$ with respect to  the text $T = [0 \ldots n)$ starting from the last element of the pattern (line 2).
Moreover, we keep track of the first $\LCP$ values with respect to the maximal match of length $1$ (line \ref{lalgo:1_lcp_initialization}).

At each iteration of the loop (line~\ref{lalgo:1_MScomputation_forLoop}), the algorithm tries to extend the match backwards position by position.
If the match can be extended (line~\ref{lalgo:1_MScomputation_match}), then we use Algorithm~\ref{algo:MSmatch} to compute the entry of the \eMS. Otherwise, we use Algorithm~\ref{algo:MSMismatch} to compute the next entry of \eMS (line \ref{lalgo:1_MScomputation_mismatch}).

\begin{algorithm}[tb]

\DontPrintSemicolon
\LinesNumbered
\Input{Pattern $P[0,m)$}
\Output{Extended matching statistics $\eMS[0..m)$}
\BlankLine
 $q \leftarrow \BWT.\Select_{P[m-1]}(1)$\;
$\eMS[m-1] \leftarrow (\Pos:\SA[q] - 1, \Len:1, \Twice: 1)$ \label{lalgo:1_eMS_inizialization}\;
$lcp_p \leftarrow 0$, $lcp_s \leftarrow 1$ \label{lalgo:1_lcp_initialization}\;
$q \leftarrow \LF(q)$\;
\For{$i\leftarrow m-2$ \Down \KwTo $0$}{ \label{lalgo:1_MScomputation_forLoop}
    \uIf{$\BWT[q] = P[i]$}{
        $\eMS[i], lcp_p, lcp_s \leftarrow$ MSMatch$(P[i], q, \eMS[i+1].\Pos, \eMS[i+1].\Pos, lcp_p, lcp_s)$ \label{lalgo:1_MScomputation_match}
    }
    \Else{
        $\eMS[i], lcp_p, lcp_s \leftarrow$ MSMisMatch$(P[i], q, \eMS[i+1].\Pos, \eMS[i+1].\Pos, lcp_p, lcp_s)$\label{lalgo:1_MScomputation_mismatch}
    }
    $q \leftarrow \LF(q)$\;
}
\Return{\eMS}
\caption{Computation of \eMS.\label{algo:MScomputation}}
\end{algorithm}

\paragraph*{Match case}

Suppose $\eMS[i+1 \ldots m)$ has already been processed and that $P[i] = T[\eMS[i+1].\Pos-1]$, namely we can further extend the longest match at the previous step by one position to the left.
Algorithm~\ref{algo:MSmatch} handles such scenario. 

Let $q$ be such that $\SA[q] = \eMS[i+1].\Pos-1$.
Hence, we have that $\eMS[i].\Pos = \eMS[i+1].\Pos-1$ and $\eMS[i].\Len = \eMS[i+1].\Len+1$ (line \ref{lalgo:2_extend_pos_len}).
At this point, we search for the greatest index $q_p$ among those smaller than $q$ such that $\BWT[q_p] = P[i]$.  
As discussed before, when $q_p = q-1$, then  $\LCP[\LF(q)] = \LCP[q]+1 = lcp_p +1$ (line \ref{lalgo:2_extend_lcp_p}).
Otherwise we can compute the $\LCE$ query between $\SA[q]$ and $\SA[q_p]$, to which we add $1$ for the match with $P[i]$ in correspondence of $\BWT[q]$ and $\BWT[q_p]$ (line \ref{lalgo:2_lcp_p_lce}).
Note that $\SA[q] = \eMS[i+1].\Pos$, while $q_p$ is the last index of its equal-letter run (and therefore $\SA[q_p]$ is stored).

Analogously we compute $lcp_s$ (lines \ref{lalgo:2_extend_lcp_s}-\ref{lalgo:2_lcp_s}) and, by Lemmas \ref{le:twice} and \ref{le:lcp_lce+1}, we assign to $\eMS[i].\Twice$ the maximum between $lcp_p$ and $lcp_s$.

\begin{algorithm}[tb]
\LinesNumbered
\DontPrintSemicolon
\BlankLine
\nl $\Pos\leftarrow \eMS[i+1].\Pos - 1, \Len \leftarrow \eMS[i+1].\Len + 1$\; \label{lalgo:2_extend_pos_len}
\nl $c \leftarrow \BWT.\Rank_{P[i]}(q)$\;
\nl \uIf{$\BWT[q-1] = P[i]$}{
    $lcp_p \leftarrow lcp_p + 1$ \label{lalgo:2_extend_lcp_p}
}
\nl \Else{
    \nl $q_p \leftarrow \BWT.\Select_{P[i]}(c)$\;
    \nl $lcp_p \leftarrow \Min(lcp_p, \LCE(\eMS[i+1].\Pos, \SA[q_p])) + 1$ \label{lalgo:2_lcp_p_lce}\;
}
\nl \uIf{$\BWT[q+1] = P[i]$}{
    $lcp_s \leftarrow lcp_s + 1$ \label{lalgo:2_extend_lcp_s}
}
\nl \Else{
    \nl $q_s \leftarrow \BWT.\Select_{P[i]}(c+2)$\;
    \nl $lcp_s \leftarrow \Min(lcp_s, \LCE(\eMS[i+1].\Pos, \SA[q_s])) + 1$ \label{lalgo:2_lcp_s}\;
}
\nl $\Twice \leftarrow \Max(lcp_p, lcp_s)$\;
\nl\Return{ $(\Pos, \Len, \Twice), lcp_p, lcp_s$ }\;
\caption{MSMatch$(P[i], q, \eMS[i+1].\Pos , \eMS[i+1].\Len, lcp_p, lcp_s$) \label{algo:MSmatch}}
\end{algorithm}
 
\paragraph*{Mismatch case}

We use Algorithm~\ref{algo:MSMismatch} when $q$ is such that $\BWT[q] \neq P[i]$.
We search for the index $q'$ in $\SA$ such that, among the suffixes of $T$ preceded by $P[i]$, at position $\SA[q']$ in $T$ starts the longest match with a prefix of $P[i+1..m)$.
Note that $T[\SA[q']-1] = P[i]$, and that $q'$ is either $q_p$ or $q_s$.

Hence, if $q_p = q - 1$, then by Lemma \ref{le:lcp_lce+1} the longest common prefix of $T[\SA[q']..n)$ and $P[i+1..m)$ has length $lcp'_p = lcp_p$ computed at the previous step (line \ref{lalgo:3_lcp_p+1}), otherwise we compute and store the $\LCE$ between $T[q..n)$ and $T[q_p..n)$ (line \ref{lalgo:3_lcp_p_lce}).
A symmetric procedure is used to compute $lcp'_s$ (lines \ref{lalgo:3_start_lcp_s}-\ref{lalgo:3_end_lcp_s}).

Without loss of generality, we assume that $lcp'_s \geq lcp'_p$, hence  $\eMS[i].\Pos = \SA[q_s] - 1$.
Then $\eMS[i].\Len = lcp'_s + 1$ and $lcp_p = lcp'_p + 1$ (line \ref{lalgo:3_len_lcp_s}).
We add $1$ to both $lcp'_s$ and $lcp'_p$ because both matches can be extended by one position on the left since $P[i] = \BWT[q_p] = \BWT[q_s]$. 
In order to compute $\eMS[i].\Twice$ we need to compute the value of $lcp_s$ with respect to $q_s$. 
To do so, we look for the smallest index $q'_s$ greater than $q_s$ such that $\BWT[q'_s] = P[i]$, and then apply a similar procedure to Algorithm 2 (lines \ref{lalgo:3_start_algo_2}-\ref{lalgo:3_end_algo_2}).
In this case, if $\BWT[q_s+1] = P[i]$, then we can retrieve $lcp_s$ from $\LCP[q_s+1]$ since $q_s$ is in correspondence of a run boundary.
Symmetrically we handle the case $lcp'_p > lcp'_s$ (lines \ref{lalgo:3_start_symmetric}-\ref{lalgo:3_end_symmetric}).
Finally, we compute $\eMS[i].\Twice$ by picking the maximum between $lcp_p$ and $lcp_s$.

\begin{theorem}
\label{th:mums_space_time_bounds}
Given a text $T[0..n)$, we can build a data structure in $\Oh(r+g)$ space that allows to compute the set $\MUMs$ between any pattern $P[0..m)$ and $T$ in $\Oh(m \cdot (t_{\LF} + t_{\LCE} + t_{\pred}))$ time.
\end{theorem}
\begin{proof}
Algorithm~\ref{algo:MScomputation}, Algorithm~\ref{algo:MSmatch} and Algorithm~\ref{algo:MSMismatch} show how to compute the $\eMS$ array in $m$ steps by using the data structure used in \cite{DBLP:conf/dcc/BoucherGIKLMNP021} of size $\Oh(r+g)$, to which we add $\Oh(r)$ words from the $\LCP$ array, preserving the space bound.
Since at each step the dominant cost depends on the $\LF$, $\LCE$, and $\Rank/\Select$ queries, \eMS is computed in $\Oh(m(t_{\LF} + t_{\LCE} + t_{\pred}))$ time. 
By Lemmas \ref{le:unique_in_T} and \ref{le:maximality}, we can build the set $\Lh$ in $\Oh(m)$ steps from the $\eMS$ array.
Recall that $\Lh$ contains those indices $i\in [0..m)$ such that $P[i..i+\eMS[i].\Len)$ is a maximal match that occurs only once in $T$.

Now we have to search those indices in $\Lh$ that are also unique in $P$.
A simple algorithm is to build both the $\LCP$ and $\ISA$ array of $P$, and then check for each $i \in \Lh$ if both $\LCP[\ISA[i]]$ and $\LCP[\ISA[i]+1]$ (or only  $\LCP[\ISA[i]]$ if $\ISA[i] = m$) are smaller than $\eMS[i].\Len$, i.e. the same property that we use to check the uniqueness in $T$.
Both structures can be build in $\Oh(m)$ time. 
The overall time is $\Oh(m(t_{\LF}+t_{\LCE}+t_{\pred})+m+m)$, which collapses to $\Oh(m(t_{\LF}+t_{\LCE}+t_{\pred}))$.
\end{proof}

Note that both $g$ and $t_{\LCE}$ depends on the grammar scheme chosen.
In fact, if exists a data structure of size $\lambda$ that supports $\LCE$ queries on a text $T$, then we can still compute $\MUMs$ in $\Oh(r+\lambda)$ space and $\Oh(m \cdot (t_{\LF} + t_{\LCE} + t_{\pred}))$ time, with $t_{\LCE}$ that depends on the data structure used.

\begin{algorithm}[tb]
\LinesNumbered
\DontPrintSemicolon
\BlankLine
\nl $c \leftarrow \BWT.\Rank_{P[i]}(q)$\;
\nl $q_p \leftarrow \BWT.\Select_{P[i]}(c)$\;
\nl $q_s \leftarrow \BWT.\Select_{P[i]}(c+1)$\;
\nl \uIf{$q_p = q-1$}{
    \nl $lcp'_p \leftarrow lcp_p$ \label{lalgo:3_lcp_p+1}
}
\nl \Else{
    \nl $lcp'_p \leftarrow \Min(\eMS[i+1].\Len, \LCE(\eMS[i+1].\Pos, \SA[q_p]))$ \label{lalgo:3_lcp_p_lce}
}
\nl \uIf{$q_s = q+1$ \label{lalgo:3_start_lcp_s}}{
    \nl $lcp'_s \leftarrow lcp_s$
}
\nl \Else{
    \nl $lcp'_s \leftarrow \Min(\eMS[i+1].\Len, \LCE(\eMS[i+1].\Pos, \SA[q_s]))$ \label{lalgo:3_end_lcp_s}
}
\nl \uIf{$lcp'_p \leq lcp'_s$}{
    \nl $\Pos \leftarrow \SA[q_s] - 1, \Len \leftarrow lcp'_s + 1$, $lcp_p \leftarrow lcp'_p + 1$ \label{lalgo:3_len_lcp_s}\;
    \nl $q'_s \leftarrow \BWT.\Select_{P[i]}(c+2)$ \label{lalgo:3_start_algo_2}\;
    \nl \uIf{$q'_s = q_s + 1$}{
        \nl $lcp_s \leftarrow \Min(\Len, \LCP[q_s + 1] + 1)$ 
    }
    \nl \Else{
        \nl $lcp_s \leftarrow \Min(\Len, \LCE(\SA[q_s], \SA[q'_s]) + 1)$ \label{lalgo:3_end_algo_2}
    }
    \nl $q \leftarrow q_s$\;
}
 \nl \Else{ \label{lalgo:3_start_symmetric}
    \nl $\Pos \leftarrow \SA[q_p] - 1, \Len \leftarrow lcp_p$, $lcp_s \leftarrow lcp'_s + 1$\;
    \nl $q'_p \leftarrow \BWT.\Select_{P[i]}(c-1)$\;
    \nl \uIf{$q'_p = q_p - 1$}{
        \nl$lcp_p \leftarrow \Min(\Len, \LCP[q_p] + 1)$ 
    }
    \nl \Else{
        \nl $lcp_p \leftarrow \Min(\Len, \LCE(\SA[q_p], \SA[q'_p]) + 1)$ \label{lalgo:3_end_symmetric}
    }
    \nl $q \leftarrow q_p$
}
\nl $\Twice \leftarrow \Max(lcp_p, lcp_s)$\;
\nl \Return{$(\Pos, \Len, \Twice), lcp_p, lcp_s$}
\caption{MSMismatch$(P[i], q, \eMS[i+1].\Pos , \eMS[i+1].\Len, lcp_p, lcp_s)$ \label{algo:MSMismatch}}

\end{algorithm}

\subsection{Computing \MUMs from \eMS}
\label{ss:mums_from_ems}

\begin{algorithm}[tb]

\LinesNumbered
\DontPrintSemicolon
\Input{Extended Matching Statistics $\eMS[0,m)$}
\Output{\MUMs}
\BlankLine
$\mathcal{L}, \MUMs \leftarrow \emptyset$\;
\For{$i \leftarrow 0$ \KwTo $m-1$}{
    \If{$(i = 0$ \Or $\MS[i-1].\Len \leq \MS[i].\Len) \textbf{ and } \MS[i].\Len > \MS[i].\Twice$ \label{lalgo:4_start_build_L}}{
        $\mathcal{L}.\add(i)$ \label{lalgo:4_end_build_L}}
}
\SortByPosition($\Lh$)\; 
$(p,\ell) \leftarrow (\eMS[\mathcal{L}[0]].\Pos, \eMS[\mathcal{L}[0]].\Len)$\;
$\Unique \leftarrow \True$\;
\For{$i \leftarrow 1$ \KwTo $|\Lh|-1$}{
    $(p',\ell') \leftarrow (\eMS[\mathcal{L}[i]].\Pos, \eMS[\mathcal{L}[i]].\Len)$\;
    \uIf{$p = p'$ }{
        \uIf{$\ell = \ell'$ }{
            $\Unique \leftarrow \False$
        }
        \ElseIf{$\ell < \ell'$ }{
            $\ell \leftarrow \ell'$\;
            $\Unique \leftarrow \True$
        }
    }
    \ElseIf{$\ell < \ell' + (p' - p)$ }{
        \If{$\Unique$}{$\MUMs.\add((p, \ell))$}
        $(p, \ell) \leftarrow (p', \ell')$\;
        $\Unique \leftarrow \True$
    }
}
\If{$\Unique$ }{
$\MUMs.\add((p, \ell))$\;
}
\Return{\MUMs}\;
\caption{retrieveMUMs(\eMS) \label{algo:retrieveMums}
}
\end{algorithm}

Here we present a different approach to compute the MUMs from the \eMS from the one in Theorem~\ref{th:mums_space_time_bounds}, that is of more practical use, and that does not require sorting the suffixes of $P$. We summarize this approach in Algorithm~\ref{algo:retrieveMums}.

Let ${\cal L}$ be the set of indexes $i \in [0..m)$ such that $P[i..\eMS[i].\Len) = T[\eMS[i].\Pos..\eMS[i].\Pos+\eMS[i].\Len)$ is a maximal and unique match in $T$.
By Lemmas \ref{le:unique_in_T} and \ref{le:maximality}, we can check in constant time if an index $i$ belongs to $\Lh$.
Note that building $\Lh$ (lines \ref{lalgo:4_start_build_L}-\ref{lalgo:4_end_build_L}) can be also executed in streaming while computing the $\eMS$ array (for simplicity of exposition of the algorithms we have separated the procedures).
Observe that a match $P[i..i+\eMS[i].\Len)$ such that $i \in \mathcal{L}$ is a \MUM if and only if it is not fully contained into another candidate, i.e. it does not exist $j \in \mathcal{L}\setminus\{i\}$ such that (i) $\eMS[j].\Pos \leq \eMS[i].\Pos$ and (ii) $\eMS[i].\Pos + \eMS[i].\Len \leq \eMS[j].\Pos + \eMS[j].\Len$ (Theorem \ref{th:mum}).
Hence, we sort the elements in $\Lh$ with respect to the position in $T$, and starting from $\mathcal{L}[0]$, we compare every entry with the following and if both factors are not contained into the other, we store in the set $\MUMs$ the one with the smallest starting position and keep track of the other one, otherwise we simply discard the one that is repeated and continue with the following iteration.

To handle the special case when two candidates $i \neq j\in \Lh$ are such that $T[\eMS[i].\Pos..\eMS[i].\Pos+\eMS[i].\Len) = T[\eMS[j].\Pos..\eMS[j].\Pos+\eMS[j].\Len)$, we further keep track whether the current maximal match is unique.
This final procedure, excluding the building time for $\Lh$ that is done in streaming, takes $\mathcal{O}(|\mathcal{L}|\log |\mathcal{L}|)$ time, since the sorting of the indexes in $\Lh$ dominates the overall cost.

\section{Experimental results}

We implemented our algorithm for computing MUMs and measured its performances on real biological datasets. We performed the experiments on a desktop computer equipped with 3.4\,GHz Intel Core i7-6700 CPU, 
8\,MiB L3 cache. and 16\,GiB of DDR4 main memory. The machine had no other significant CPU tasks running, and only a single thread of execution was used.
The OS was Linux (Ubuntu 16.04, 64bit) running kernel 4.4.0. All programs were compiled using {\tt gcc} version 8.1.0 with {\tt-O3} {\tt-DNDEBUG} {\tt -funroll-loops} {\tt -msse4.2} options. We recorded the runtime and memory usage using the wall clock time, CPU time, and maximum resident set size from \texttt{/usr/bin/time}.

\paragraph*{Setup}
We compare our method (\ours) with MUMmer~\cite{Marcais2018-ys} ({\tt mummer}). We tested two versions of {\tt mummer}, v3.27~\cite{Kurtz2004-cf} ({\tt mummer3}) and v4.0~\cite{Marcais2018-ys} ({\tt mummer4}). We executed {\tt mummer} with the {\tt -mum} flag to compute MUMs that are unique in both the text and the pattern, {\tt-l 1} to report all MUMs of length at least 1, and {\tt -n} to match only A,C,G,and T characters. We setup \ours to produce the same output as {\tt mummer}.
We did not test against Mauve~\cite{Darling2010-jw} because the tool does not directly reports MUMs. We also did not consider algorithms that does not produces an index for the text that can be queried with different patterns without reconstructing the index, e.g. the algorithm described in M{\"a}kinen et al.~\cite[Section 11.1.2]{makinen2015genome}. The experiments that exceeded exceeded 16 GB of memory were omitted from further consideration.

\paragraph*{Datasets}

\captionsetup[subtable]{position=bottom}
\begin{table}\centering
\subfloat[Collections of chromosome 19.\label{tab:chr19}]{
\begin{tabular}{ | r | r r |}
\hline
No. seqs & $n$ (MB)	& $n/r$ \\ 
\hline
\hline
1 & \num{59} & 1.92 \\
2 & \num{118} & 3.79 \\
4 & \num{236} & 7.47 \\
8 & \num{473} & 14.78 \\
16 & \num{946} & 29.19 \\
32 & \num{1892} & 57.63 \\
64 & \num{3784} & 113.49 \\
128 & \num{7568} & 222.23 \\
256 & \num{15136} & 424.93 \\
512 & \num{30272} & 771.53 \\
\hline
\end{tabular}
}
\hspace{4em}
\subfloat[Collections of SARS-CoV2 genomes.\label{tab:sars}]{
\begin{tabular}{ | r | r r | }
\hline
No. seqs & $n$ (MB)	& $n/r$ \\ 
\hline
\hline
\num{1562}	& \num{46}	& 459.57 \\ 
\num{3125}  & \num{93}	& 515.42 \\ 
\num{6250} & \num{186}	& 576.47 \\ 
\num{12500} & \num{372}	& 622.92 \\ 
\num{25000} & \num{744}	& 704.73 \\ 
\num{50000} & \num{1490}   & 848.29 \\ 
\num{100000} & \num{2983}	& 1060.07 \\ 
\num{200000} & \num{5965}	& 1146.24 \\ 
\num{300000} & \num{8947}	& 1218.82 \\ 
\hline
\end{tabular}
}
\caption{\label{tab:n/r} Dataset used in the experiments. For each collection of datasets of the human chromosome 19 ({\tt chr19}) dataset in Table~\ref{tab:chr19} and for the SARSCoV2 ({\tt sars-cov2}) dataset in Table~\ref{tab:sars}, we report the number of sequences (No. seqs), the length $n$ in Megabytes (MB), and the ratio $n/r$, where $r$ is the number of runs of the \BWT\ for each number of sequences in a collection .}
\end{table}

We evaluated our method using real-world datasets. We build our index for up to 512 haplotypes of human chromosome 19 from the 1000 Genomes Project~\cite{1000genomes} and up to 300,000 SARS-CoV2 genomes from EBI's COVID data portal~\cite{ebi21}. We provide a complete list of accession numbers in the repository. We divide the sequences into 11 collections of 1, 2, 3, 4, 8, 16, 32, 64, 128, 256, 512 chromosomes 19 (\texttt{chr19}) and 9 collections of 1,562, 3,125, 6,250, 1250,00, 25,000, 50,000, 100,000, 200,000, 300,000 genomes of SARS-CoV2 (\texttt{sars-cov2}). In both datasets, each collection is a superset of the previous one. In Table~\ref{tab:n/r} we report the length $n$ of each collection and the ratio $n/r$, where $r$ is the number of runs of the \BWT.

Furthermore, for querying the datasets, we used the first haplotype of chromosome 19 of the sample NA21144 from the 1000 Genomes Project, and the genome with accession number MZ477765 from EBI's COVID data portal~\cite{ebi21}.

\paragraph*{Results}

In Figure~\ref{fig:chr19} we show the construction and query time and space for \ours\ and {\tt mummer}. Since {\tt mummer} is not able to decouple the construction of the suffix tree from the query, for our method we report the sum of the running times for construction and query, and the maximum resident set size of the two steps. We observe that on {\tt chr19} {\tt mummer3} is up to 9 times faster than \ours, while using up to 8 times more memory, while {\tt mummer4} is up to 19 times faster than \ours, while using up to 7 times more memory. However both {\tt mummer3} and {\tt mummer4} cannot process more than 8 haplotypes of \texttt{chr19} due to memory limitations. \ours\ was able to build the index and query in 48 minutes for 512 haplotypes of \texttt{chr19} while using less than 11.5 GB of RAM. On {\tt sars-cov2}, {\tt mummer3} is up to 6.5 times faster than \ours, while using up to 24 times more memory, while {\tt mummer4} is up to 1.2 times slower than \ours, while using up to 25 times more memory. {\tt mummer3} was not able to process more than 25,000 genomes while {\tt mummer4} were not able to query mote than 12,500 genomes of \texttt{sars-cov2} due to memory limitations.

\begin{figure}[t!]
    \centering
    \subfloat[Construction time {\tt chr19}.]{\label{fig:chr19:time}
 	\centering
 		\includegraphics[width=0.49\textwidth]{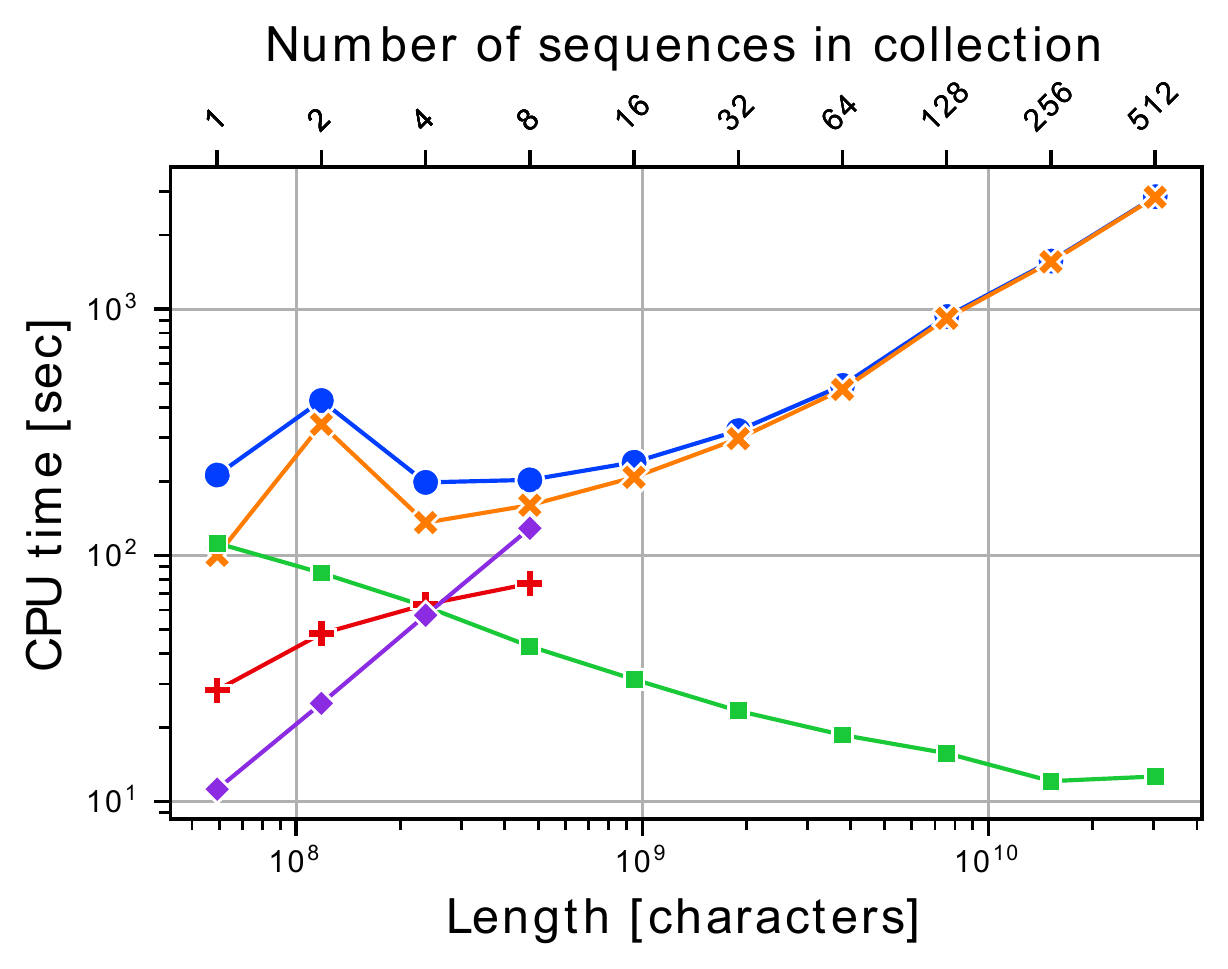}
    }~
   \subfloat[Peak memory {\tt chr19}.]{\label{fig:chr19:space}
     \centering
     		\includegraphics[width=0.49\textwidth]{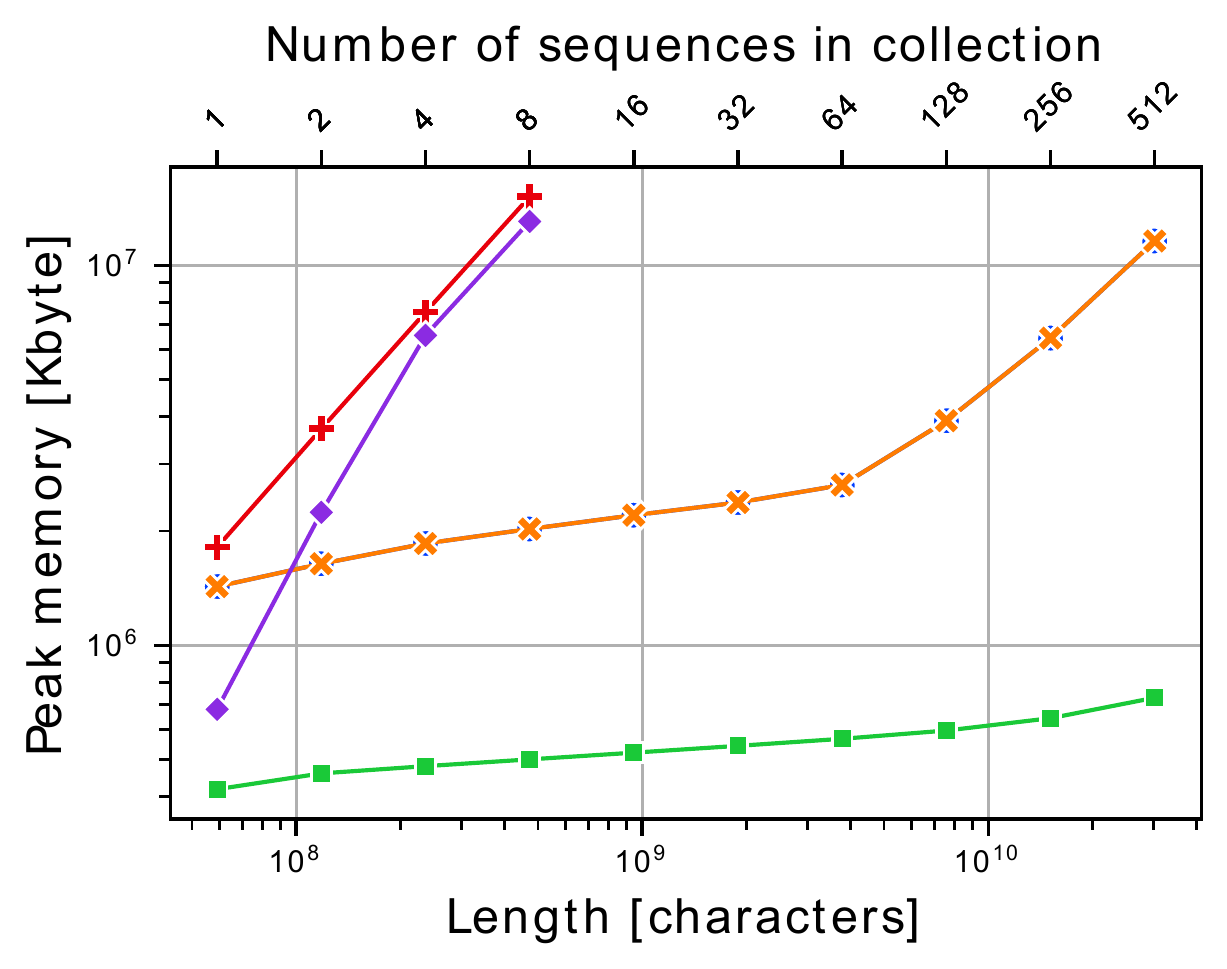}
    }\\
    \subfloat[Construction time {\tt sars-cov2}.]{\label{fig:sars:time}
 	\centering
 		\includegraphics[width=0.49\textwidth]{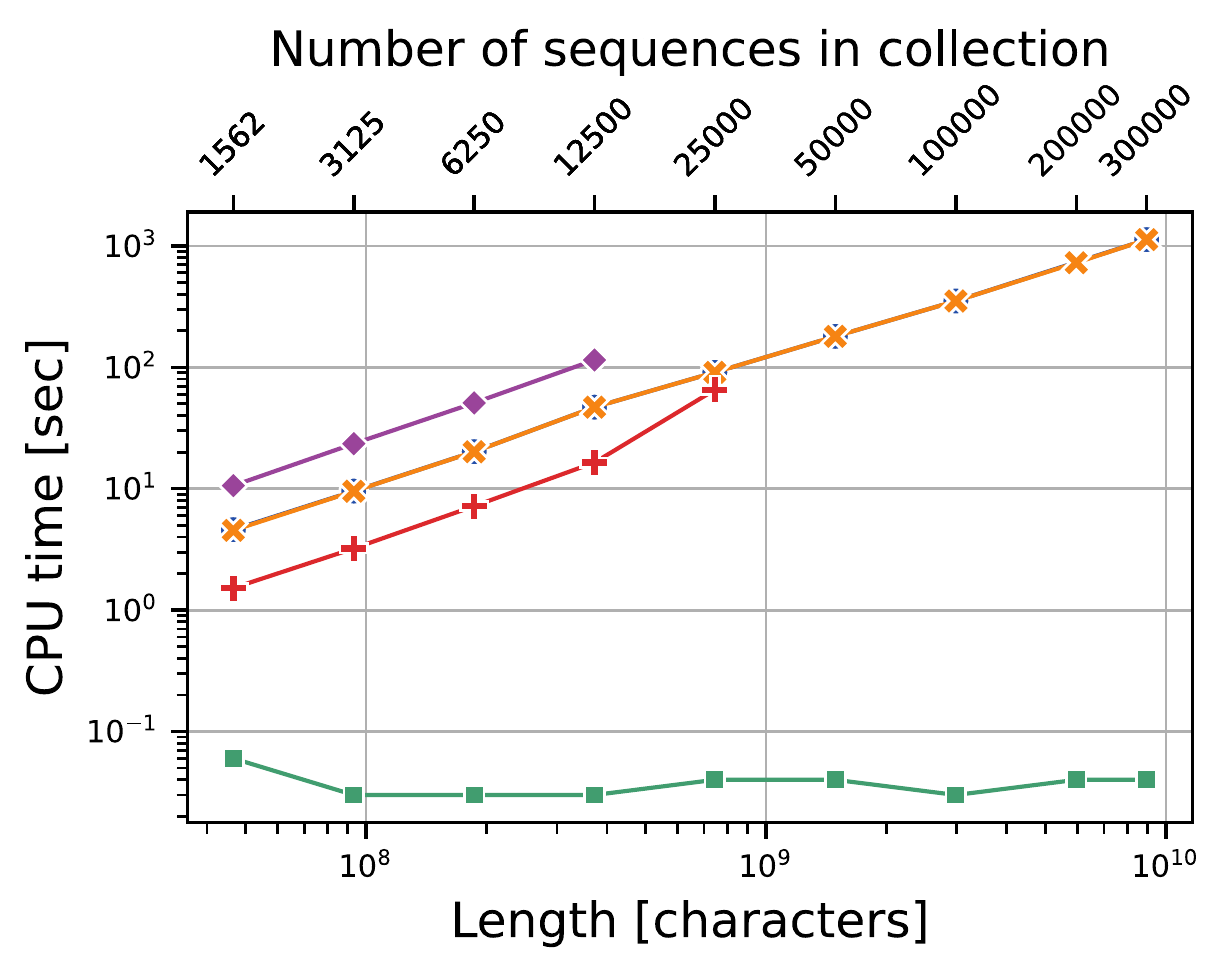}
    }~
   \subfloat[Peak memory {\tt sars-cov2}.]{\label{fig:sars:space}
     \centering
     		\includegraphics[width=0.49\textwidth]{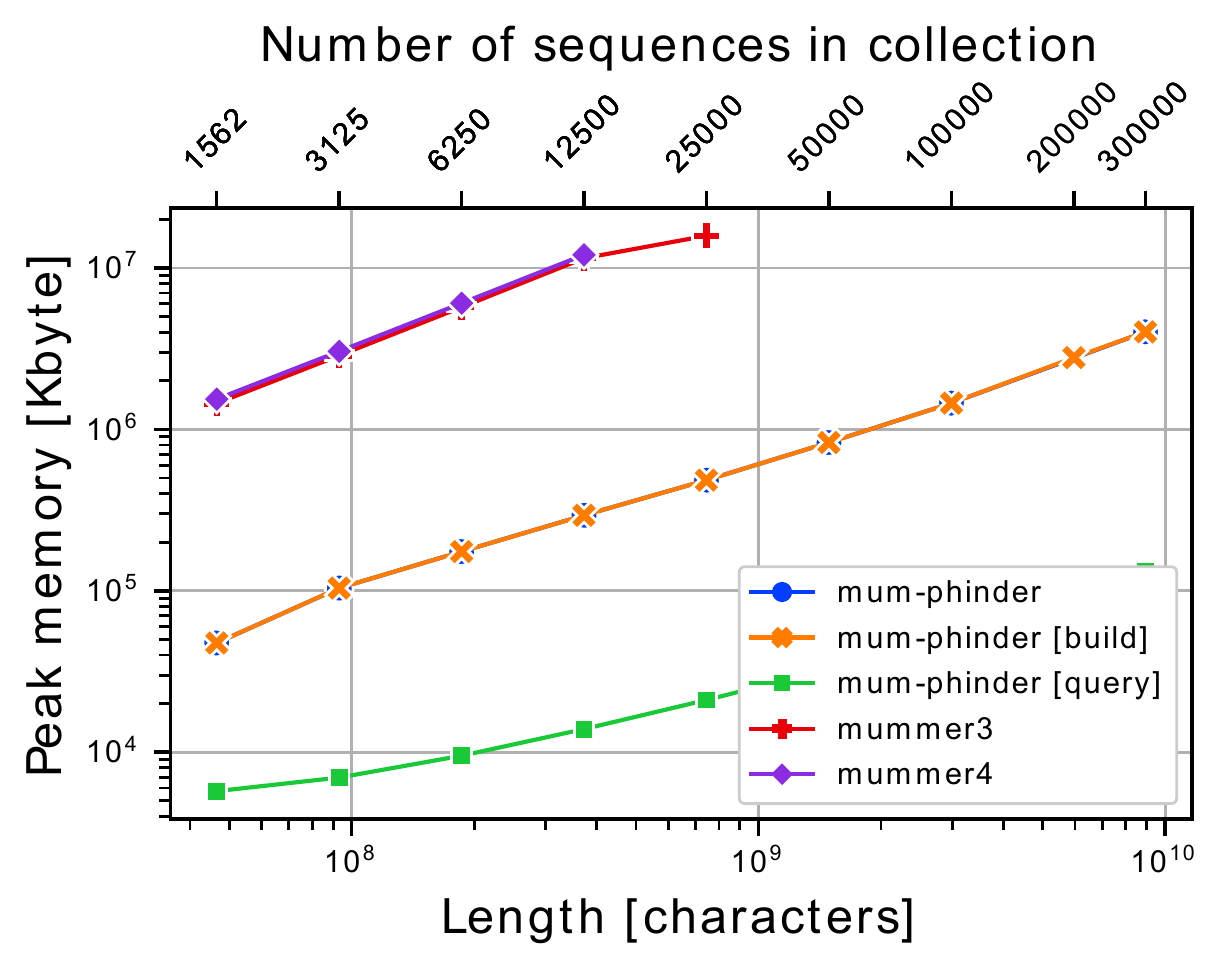}
    }
\caption{Human chromosome 19 and SARS-CoV2 genomes dataset construction CPU time and peak memory usage. We compare \ours{} with {\tt mummer3} and {\tt mummer4}. For \ours\ we report a breakdown of the construction (build) and query time and space. Note that for \ours\ we consider as time the sum of construction and query time, while for memory we consider the maximum between construction and query memory.\label{fig:chr19}}
\end{figure}

In Figure~\ref{fig:chr19} we also show the construction time and space for \ours. We observe that the construction time grows with the number of sequences in the dataset, however the query time decreases while increasing the number of sequences in the index with a 9x speedup when moving from 1 to 512 haplotypes of {\tt chr19}. A similar phenomenon is observed in~\cite{DBLP:conf/dcc/BoucherGIKLMNP021} and it is attributed to the increase number of match cases (Algorithm~\ref{algo:MSmatch}) while increasing the number of sequences in the index. From our profiling (data not shown) the more time-demanding part of the queries are \LCE\ queries, which are not performed in case of matches. This observation also motivates the increase in the control logic of Algorithm~\ref{algo:MSMismatch} to limit the number of \LCE\ queries to the essential ones.

\end{document}